\documentclass[twoside,leqno,twocolumn]{article}

\usepackage[letterpaper]{geometry}

\usepackage{ltexpprt}
\usepackage{amssymb}
\usepackage{amsmath}
\usepackage{graphicx}
\usepackage{booktabs}
\usepackage{multirow}
\DeclareMathOperator*{\argmax}{arg\,max}

\usepackage{color}

\newtheorem{remark}{Remark}[section]
\newcommand{\kibitz}[2]{\ifnum\Comments=0\textcolor{#1}{#2}\fi}

\begin{document}

\title{On the Request-Trip-Vehicle Assignment Problem}
\author{J. Carlos Mart\'{i}nez Mori\thanks{Center for Applied Mathematics, Cornell University. Email: jm2638@cornell.edu. Work partially supported by the Dwight David Eisenhower Transportation Fellowship Program under Award No. 693JJ32145020.}
\and Samitha Samaranayake\thanks{School of Civil and Environmental Engineering, Cornell University. Email: samitha@cornell.edu. Work partially supported by the National Science Foundation under Grant No. CNS-1952011 and DMS-1839346.} }

\date{}

\maketitle


\fancyfoot[R]{\scriptsize{Copyright \textcopyright\ 2021 by SIAM\\
Unauthorized reproduction of this article is prohibited}}





\begin{abstract}
\small\baselineskip=9pt 
The \textit{request-trip-vehicle} assignment problem is at the heart of a popular decomposition strategy for online vehicle routing. In this framework, assignments are done in batches in order to exploit any shareability among vehicles and incoming travel requests. We study a natural ILP formulation and its LP relaxation. Our main result is an LP-based randomized rounding algorithm that, whenever the instance is feasible, leverages mild assumptions to return an assignment whose: \textit{i}) expected cost is at most that of an optimal solution, and \textit{ii}) expected fraction of unassigned requests is at most $1/e$. If trip-vehicle assignment costs are $\alpha$-approximate, we pay an additional factor of $\alpha$ in the expected cost. We can relax the feasibility requirement by considering the penalty version of the problem, in which a penalty is paid for each unassigned request. We find that, whenever a request is repeatedly unassigned after a number of rounds, with high probability it is so in accordance with the sequence of LP solutions and not because of a rounding error. We additionally introduce a deterministic rounding heuristic inspired by our randomized technique. Our computational experiments show that our rounding algorithms achieve a performance similar to that of the ILP at a reduced computation time, far improving on our theoretical guarantee. The reason for this is that, although the assignment problem is hard in theory, the natural LP relaxation tends to be very tight in practice. 
\end{abstract}

\section{Introduction}
In the \textit{request-trip-vehicle} (RTV) assignment problem, we are given a set $R$ of travel requests, a set $T$ of candidate trips (i.e., a collection of subsets of $R$), and a set $V$ of vehicles. Assigning a vehicle to a trip has an associated cost, typically representing distance traveled or incurred delays. The problem is to find a minimum cost set of trip-vehicle assignments such that: \textit{i}) each request appears in exactly one trip-vehicle assignment, and \textit{ii}) each vehicle is assigned to at most one trip.

The problem is at the heart of a decomposition strategy for online vehicle routing problems popularized by Alonso-Mora et al.~\cite{alonso2017demand} within the context of high-capacity ridesharing. Compared to traditional literature on vehicle routing~\cite{toth2014vehicle}, this \textit{anytime} optimal framework decouples the routing and matching aspects of the problem, making it well-suited for parallel, online computation. The term \textit{online} refers to the real-time nature of the system. Assignments are done in batches (e.g., every 5 to 30 seconds), rather than sequentially, to exploit any shareability~\cite{santi2014quantifying} among vehicles and incoming travel requests. This style of solution approach is in fact used in practice by some well-known on-demand mobility service providers (for example, see~\cite{reengineering2017data}).

The framework exploits two key structural properties: \textit{i}) there are tight quality of service constraints (e.g., maximum wait time, maximum travel time)~\cite{santi2014quantifying,guo2018solving,xu2020geoprune}, and \textit{ii}) the feasible space is downward closed. In particular, a necessary condition for a potential trip-vehicle assignment to be feasible is that all of its sub-trip-vehicle assignments are feasible. This means $T$ itself is downward closed. Together, these properties help prune a priori infeasible trips and vehicle assignments, thereby thwarting the combinatorial explosion.


\subsection{Related Work.} 

Previous work around the RTV assignment problem has primarily focused on experimental performance. Ota et al.~\cite{ota2016stars} give a greedy assignment algorithm based on a request indexing scheme. Simonetto et al.~\cite{simonetto2019real} solve the problem via linear assignments. Rather than tackling the problem all at once, they decompose it into a sequence of semi-matching linear programs in which each vehicle can be matched to up to one request. In their experiments, they achieve a system performance similar to that of~\cite{alonso2017demand} in about a fourth of the time. Lowalekar, Varakantham, and Jaillet~\cite{lowalekar2019zac} generate assignments at the \textit{zone path} level. They group trips that have compatible pickup and drop off locations. Riley, Legrain, and Van Hentenryck~\cite{riley2019column} propose a column generation algorithm under soft constraints, where quality of service is enforced through a Lagrangian approach. They solve their pricing problem through an anytime optimal algorithm that, similar to the trip generation step in~\cite{alonso2017demand}, explores the feasible space in increasing order of trip size. Their algorithm takes exponential time in the worst case, but their experiments suggest soft constraints can reduce wait times and route deviations. 

Bei and Zhang~\cite{bei2018algorithms} show the problem is NP-hard even when no more than two requests can share a vehicle. They also give a $5/2$-approximation algorithm for the total distance minimization version of this special case (under the additional assumption that there are \textit{exactly} twice as many requests as there are vehicles). Li, Li, and Lee~\cite{li2020trip} match the $5/2$ approximation guarantee while relaxing the latter assumption. Namely, they only require there being \textit{at most} twice as many requests as there are vehicles. Luo and Spieksma~\cite{luo2020approximation} obtain a 2-approximation algorithm under the same assumptions as in~\cite{bei2018algorithms}. They also obtain a $5/3$-approximation when the objective is to minimize the total latency. Lowalekar, Varakantham, and Jaillet~\cite{lowalekar2020competitive} study the competitive online version of the problem in the special case where vehicles return to their depot after serving a set of requests. If requests arrive in batches under a known adversarial arrival distribution, they obtain an algorithm whose expected competitive ratio is 0.3176 whenever vehicles have seating capacity 2. If vehicles have seating capacity $k > 0$, their expected competitive ratio is $\gamma > 0$, where $\gamma$ is the solution to $\gamma=(1-\gamma)^{k+1}$.

\subsection{Contributions.} 
\label{sec: contributions}
\subsubsection*{1.} We consider a natural integer linear programming (ILP) formulation that, in the worst case, has exponentially many variables. It is therefore not immediately clear whether one can always solve or even approximate its linear programming (LP) relaxation in polynomial time. To this end, we study the dual separation problem (i.e., column generation). The hope would be to approximately separate over the dual to approximately solve the primal, in the style of~\cite{carr2000randomized,jain2003packing,fleischer2011tight}. 

Although our ILP formulation is more general, we focus on the typical case in which trip-vehicle assignment costs correspond to the distance traveled by a vehicle when serving the requests in a trip. In this case, we identify the core of the dual separation problem as an instance of the \textit{net-worth maximization} version of the prize-collecting traveling salesman problem (TSP). We note there is a closely related inapproximability result by Feigenbaum, Papadimitriou, and Shenker~\cite{feigenbaum2001sharing} for the net-worth maximization version of the prize-collecting Steiner tree problem. Intuitively, we expect\footnote{The results for Steiner tree do not immediately translate into results for TSP since, in this version of the problem, the objective function is the difference of two terms.} their result to carry over to the tour version, but their gadget is not easily adaptable for this purpose. To the best of our knowledge, this is an open problem. We also make a note on incompatible statements made in passing in the literature, particularly around the applicability of existing approximation algorithms for a different version of the prize-collecting TSP.

This unfavorable prospect motivates us to assume the trip list $T$ is pre-computed and polynomial-sized. While this seems a rather strong assumption at first glance, it becomes much more reasonable if we a priori prune $T$ based on fixed vehicle seating capacities $k > 0$ and tight quality of service constraints. We emphasize this is a key factor that differentiates ridesharing applications of vehicle routing (e.g., in contrast to applications in logistics), and that pre-computing $T$ is in fact what is done in practical approaches for this problem~\cite{santi2014quantifying,alonso2017demand,reengineering2017data}.

\subsubsection*{2.}

Our second and main contribution is a simple LP-based randomized rounding algorithm for the RTV assignment problem with provable performance guarantees. We summarize our result as follows.

\begin{theorem}
\label{thm: main}
Suppose we have a feasible instance of the RTV assignment problem with $T$ polynomial-sized. If trip-vehicle costs are oracle-given and monotonic increasing w.r.t. request inclusion\footnote{The cost of a trip-vehicle assignment cannot decrease by adding an extra request.}, there is a randomized algorithm such that:
    \begin{itemize}
        \item The expected cost of the final solution is at most that of an optimal solution.
        \item The expected fraction of unassigned requests is at most $1/e$ (i.e., less than $36.8\%$ of all requests).
    \end{itemize}
    If trip-vehicle costs are $\alpha$-approximate, we pay an additional factor of $\alpha$ in the expected cost.
\end{theorem}

For example, Theorem~\ref{thm: main} is applicable when trip-vehicle assignment costs correspond to distance traveled. Our rounding technique is similar in style to that of Raghavan and Thompson~\cite{raghavan1987randomized} for the minimum capacity multi-commodity flow problem. In addition, we show that \textit{i}) the bound of $1/e$ on the rejection rate is tight for our algorithm, and \textit{ii}) the integrality gap of the natural LP relaxation is at least $2$.

In practice, a feasible solution to the RTV assignment problem might not always exist (e.g., when vehicle demand exceeds supply). Therefore, in practice, we may instead need to consider the penalty version of the problem, in which a penalty is paid for each ignored request. We show that we can still use our algorithm for this version of the problem, and make the following remarks:

\begin{itemize}
    \item  
    In practice, unassigned requests are carried over to the next round of batch assignments (e.g., 5 to 30 seconds later). We argue that for any request $r \in R$, with high probability over a number of rounds, our algorithm either assigns it to a vehicle or purposely ignores it, in accordance to the sequence of LP solutions. That is to say, if a request is ignored by our algorithm after various rounds, it is increasingly likely to be so because of the LP solutions and not because of a rounding error.
    \item When the penalty terms eclipse the assignment costs (e.g., routing costs), the penalty version of the problem is essentially an instance of the cardinality version of a maximum coverage problem with group budget constraints, for which an $1/e$-approximation algorithm is known~\cite{chekuri2004maximum}, assuming $T$ is given explicitly. Our algorithm matches this guarantee in the sense that the expected fraction of requests that could have been covered but were not is at most $1/e$, while it provides the additional benefit that the quality of the assignment is also optimized as part of the LP.
\end{itemize}

To the best of our knowledge, this is the first algorithmic result with provable performance guarantees for the RTV assignment problem in its full generality. In particular, we allow high-capacity ridesharing (more than two requests can share a vehicle) under the general class of monotonic cost coefficients. Our techniques generalize to a class of set-partitioning problems, which we formalize in the Appendix.

\subsubsection*{3.} 

We complement our analysis with computational experiments. We generate 10 distinct sets of $1,440$ simulated instances of the penalty version of the RTV assignment problem, corresponding to different combinations of number of vehicles and vehicle capacities. In our experiments, we evaluate \textit{i}) solving the ILP to optimality using a commercial solver, \textit{ii}) our randomized rounding algorithm, and \textit{iii}) a deterministic rounding heuristic inspired by our randomized technique. 

We observe that the percentage of rejected requests incurred by either rounding method is consistently below one percentage point higher than it is by solving the ILP (our heuristic performs slightly better than our randomized algorithm). Moreover, our implementation of either rounding method is faster than solving the ILP, with the percent improvement for mean and median values generally ranging between $2-7\%$. The percent improvement for mean values is consistently higher than the percent improvement for median values, showing computation time improvements are skewed to the right. In other words, the percent improvement is higher for the worst case (i.e., slowest) instances, which are arguably the most critical for real-time applications.

The performance of our randomized rounding algorithm in practice is far better than the theoretical guarantee in Theorem~\ref{thm: main}, but admittedly the ILP can be solved much faster than anticipated. Therefore, we investigate this further and find that the overwhelming majority ($\sim 96\%$) of supported LP variables are in fact integral. We moreover observe that $\sim 62\%$ of non-integral LP variables are half-integral. Of course, commercial ILP solvers leverage the quality of these fractional solutions, and our rounding algorithms make little to no mistakes on these assignments.

Our experiments confirm what we know anecdotally from other researchers and practitioners: although the RTV assignment problem is hard in theory~\cite{bei2018algorithms}, the natural LP relaxation tends to be very tight. 



\subsection{Future Research.}
\label{sec: future research}

We introduce rounding algorithms that are robust to worst case fractional solutions. However, we also present empirical evidence suggesting the natural LP relaxation can be very tight in practice. Therefore, we believe it would be valuable to rigorously understand whether or not this is always the case in practice (e.g., using data sets from various cities, beyond the habitual NYC Taxi and Limousine Commission (TLC) data~\cite{taxidata} we used in our experiments). If there are instances in which the LP is not as tight, they will likely be harder commercial ILP solvers, and so they might more strongly showcase the benefit of LP rounding. Lastly, we note that in practice the bottleneck remains to be the generation of the candidate trip list $T$. The difficulty of this step corresponds to our discussion of the dual separation problem and its approximability.

\subsection{Organization.}
\label{sec: organization}

The remainder of this paper is organized as follows. In Section~\ref{sec: preliminaries} we introduce our notation and our ILP formulation. In Section~\ref{sec: on the lp relaxation} we study its LP relaxation. In Section~\ref{sec: randomized rounding} we present and analyze our randomized rounding algorithm. In Section~\ref{sec: penalty version and multiple rounds} we introduce the penalty version of the problem. Lastly, in Section~\ref{sec: computational experiments} we introduce our deterministic rounding heuristic and share our computational experience.

\section{Preliminaries}
\label{sec: preliminaries}

\subsection{Notation and Assumptions.}
\label{sec: notation and assumptions} 

Let $T(r)$ be the set of trips that contain request $r \in R$. Let $V(t)$ be the set of vehicles that can serve trip $t \in T$ and $T(v)$ be the set of trips than can be served by vehicle $v \in V$. For ease of presentation, with the exception of Section~\ref{sec: computational experiments}, we assume any vehicle can serve any trip. Then, $V(t) = V$ for each $t \in T$ and $T(v) = T$ for each $v \in V$. This simplifies the LP constraints, which further simplifies our study of the dual LP. Our algorithm does not rely on this assumption. For technical reasons, we further assume $\emptyset \in T$ and $V(\emptyset) = V$. Any cost involving the empty trip $\emptyset \in T$ represents the cost of serving the passengers currently inside the vehicle, which is zero if the vehicle is empty. We use $\text{OPT}$ to denote the cost of an optimal solution to the RTV assignment problem and $\text{LP}(\cdot)$ to denote the objective value of a linear program. An $\alpha$-approximation algorithm for a minimization problem returns a solution of cost no more than $\alpha \geq 1$ times that of an optimal solution. An $\alpha$-approximation algorithm for a maximization problem returns a solution of value at least $0 \leq \alpha \leq 1$ times that of an optimal solution.

\subsection{Problem Formulation.}
\label{sec: problem formulation}


Let $c_{tv} \geq 0$ be the cost of assigning trip $t \in T$ to vehicle $v \in V$. Again, we assume trip-vehicle costs are monotonic increasing with respect to request inclusion. 

Typically, $c_{tv}$ corresponds to the distance traveled by a single vehicle $v$ when serving the requests in $t$ (and the passengers currently inside $v$, if any). In the simplest idealized scenario, the single vehicle routing problem is an instance of the metric TSP~\cite{lawler1985traveling,applegate2006traveling}. Then, the monotonicity assumption is met~\cite{shmoys1990analyzing}. In reality, however, the single vehicle routing problem actually corresponds to some generalization of the TSP that includes service-specific constraints. For example, it may involve paths rather than tours~\cite{hoogeveen1991analysis}, pickups and deliveries~\cite{ruland1997pickup,dumas1991pickup}, time windows~\cite{dumas1991pickup,bansal2004approximation}, capacity constraints~\cite{chalasani1999approximating}, neighborhoods~\cite{garg2000polylogarithmic}, and so on. 

Now, consider the following ILP formulation which, in the worst case, has exponentially many variables due to the size of $T$. It moreover has polynomially many constraints (except for the binary constraints). Here, $x_{tv}$ is set to $1$ if trip $t \in T$ is assigned to vehicle $v \in V$.

\begin{equation}
\label{p: rtv}
\begin{array}{lr@{}ll}
\text{min}  & \displaystyle\sum_{t \in T} \displaystyle\sum_{v \in V} c_{tv}x_{tv} & &\\
\text{s.t.}  & & & \\
    & \displaystyle\sum_{(t,v) \in T(r) \times V} x_{tv} &\geq 1,      & \forall r \in R \\
    & \displaystyle\sum_{t \in T}  x_{tv} &\leq 1,      & \forall v \in V \\
    &                                            x_{tv} &\in \{0,1\},         & \forall t \in T, v \in V
\end{array}
\end{equation} 

Our objective is to minimize the total cost of trip-vehicle assignments. The first set of constraints ensure each request is served. We can introduce these as covering constraints since all sets are downward closed and costs are monotonic increasing, but we could have equivalently introduced them as set-partitioning constraints and required equality. The second set of constraints ensure each vehicle is assigned to at most one trip. Since we assume $\emptyset \in T$, we again could have required equality for these constraints. See the Appendix for our description of this problem as a weighted set-partitioning problem.

\section{On the LP Relaxation} 
\label{sec: on the lp relaxation}

We would like to use the LP relaxation of (\ref{p: rtv}) as part of our algorithm. It is therefore natural to ask whether we can solve it in polynomial time under a succinct representation of $T$. We consider two approaches.

\subsection{An Assumption on $T$.}
\label{sec: an assumption on T}

Suppose $T$ is given explicitly as a polynomial-sized set. That is, $|T| = \text{poly}(|R|, |V|)$ and so we can directly write (and solve) the LP. We argue this assumption is not too bad. Since vehicles have a fixed seating capacity $k>0$, the size of $T$ is typically $O(n^k)$. This is not always the case since a vehicle that drops off passengers along the way may serve more than $k$ requests. Nevertheless, tight quality of service constraints typically prevent large trips from being feasible (and hence can be excluded from $T$ a priori, as in~\cite{alonso2017demand}). 

Now, this assumption does not change the fact that we need to compute the coefficients $c$, which may be NP-hard. The following claim is easy to show.
\begin{lemma}
\label{lemma: approximate c}
Suppose we have a feasible instance of the RTV assignment problem with $T$ polynomial-sized. Given an $\alpha$-approximation algorithm for $c$, we can $\alpha$-approximate $\text{LP}(\text{\ref{p: rtv}})$.
\end{lemma}
\begin{proof}
Let $x$ be an optimal solution for oracle-given coefficients $c$. Let $x'$ be an optimal solution for $\alpha$-approximate coefficients $c'$. Clearly $c'(x') \leq c'(x)$ and $c'(x) \leq \alpha \cdot c(x)$, which shows $c'(x') \leq \alpha \cdot c(x)$.
\end{proof}

For example, if the underlying routing problem were the metric TSP, we could $3/2$-approximate $c$ with Christofides' algorithm~\cite{christofides1976worst} (we can use this algorithm since, for most practical purposes in road networks, we may assume we operate on a symmetric metric space~\cite{mori2019bounded}). If the underlying graph were in addition planar, we could use Klein's polynomial time approximation scheme (PTAS)~\cite{klein2005linear}. If the underlying routing problem were instead some generalization of the TSP, the argument would follow identically except we would use whatever approximation guarantee is available (which cannot improve on what is available for the TSP).

\subsection{Dual Separation.}
\label{sec: dual separation}

Since the primal LP has exponentially many variables in the worst case, it is natural to consider its dual, which has polynomially many variables but exponentially many constraints. While we remain unable to write down all dual constraints, we could in principle solve the dual LP using the ellipsoid method together with a dual separation algorithm. 

Given a solution to a partially specified dual LP, a dual separation algorithm finds a violated dual constraint, if one exists. The polynomial time solvability of the dual separation problem implies the polynomial time solvability of the primal LP. Similarly, the polynomial time approximability of the dual separation problem could
imply the polynomial time approximability of the primal LP, as in~\cite{carr2000randomized,jain2003packing,fleischer2011tight}. In any case, note that solving the dual separation problem (whether we can do it in polynomial time or not) generates dual constraints, which is equivalent to generating columns for the primal LP. Column generation is a popular strategy for solving large LPs, and is used within branch-and-price frameworks for solving large IPs~\cite{barnhart1998branch}. In this section we explore this possibility.

Consider the dual linear program. 
\begin{equation}
\label{p: drtv}
\begin{array}{lr@{}ll}
\text{max}  & \displaystyle\sum_{r \in R} y_r - \displaystyle\sum_{v \in V} z_v & &\\
\text{s.t.}  & & & \\
    & \displaystyle\sum\limits_{r \in t} y_r - z_v &\leq c_{tv},      & \forall t \in T, v \in V \\
    &                                            y_{r} &\geq 0,         & \forall r \in R \\
    &                                            z_{v} &\geq 0,         & \forall v \in V
\end{array}
\end{equation} 
Given a polynomial time separation oracle for (\ref{p: drtv}), we can use the ellipsoid method to solve it in polynomial time. Then, since the ellipsoid method only considers polynomially many constraints in (\ref{p: drtv}), only polynomially many variables need to be written in the LP relaxation of (\ref{p: rtv}). Now, (\ref{p: drtv}) may be rewritten as follows.

\begin{equation}
\label{p: d2rtv}
\begin{array}{lr@{}ll}
\text{max}  & \displaystyle\sum_{r \in R} y_r - \displaystyle\sum_{v \in V} z_v & &\\
\text{s.t.}  & & & \\
    & (y, z_v)  &\in \mathcal{P}_v,      & \forall v \in V \\
    &                                            y_{r} &\geq 0,         & \forall r \in R \\
    &                                            z_{v} &\geq 0,         & \forall v \in V
\end{array}
\end{equation} 
Here, $\mathcal{P}_v$ is a polytope associated with vehicle $v \in V$ that is given by the constraints $\sum_{r \in t} y_r - z_v \leq c_{tv}$ for all $t \in T$ together with non-negativity constraints. Therefore, we may design a polynomial time separation oracle for (\ref{p: drtv}) by designing a polynomial time separation oracle for $\mathcal{P}_v$ and iterating over all $v \in V$. 

Consider any vehicle $v \in V$. We need a polynomial time subroutine that, given some $(y, z_v)$, either certifies that $(y, z_v) \in \mathcal{P}_v$ or returns a violated constraint. In other words, we need to verify that $\sum_{r \in t}y_r - c_{tv} \leq z_v$ for all $t \in T$. The core of our separation problem for $\mathcal{P}_v$ is then
\begin{align}
\label{eq: separation}
    \max_{t \in T} \left\{ \sum_{r \in t} y_r - c_{tv} \right\}.
\end{align}
This is because this quantity is bounded by $z_v$ if, and only if, $(y, z_v) \in \mathcal{P}_v$. We note that had we not made the simplifying assumption that any vehicle can serve any trip, the problem for vehicle $v \in V$ would need to optimize over $T(v) \subseteq T$. Presumably, $|T(v)| \ll |T|$ under tight quality of service constraints\footnote{Assuming $T(v)$ is polynomial-sized is the same as assuming $T$ is polynomial-sized, and so we have already considered this in Section~\ref{sec: an assumption on T}. This is because there are only polynomially many vehicles and $\bigcup_{v \in V} T(v) = T$ assuming our instance is feasible.}. In fact, this structure is precisely exploited in~\cite{alonso2017demand,riley2019column}.

In some select cases, it may be possible to efficiently solve (\ref{eq: separation}). For example, if we somehow knew $T$ forms a matroid (e.g., a $k$-uniform matroid where $k>0$ is a fixed vehicle seating capacity) and trip-vehicle assignment costs were additive with respect to request inclusion, we could use the greedy algorithm to optimally solve the problem~\cite{lawler2001combinatorial}. However, by and large, instances of practical interest are far less structured.

Suppose the routing problem defining the cost coefficients $c$ were the TSP. Then, we would identify (\ref{eq: separation}) as an instance\footnote{Strictly speaking, we require $T = 2^{R}$ so that maximizing over $t \in T$ is the same as maximizing over $t \subseteq R$. This corresponds to the setting with lax quality of service constraints, which is what may cause $T$ to be exponential-sized to begin with.} of the \textit{net-worth maximization} version of the prize-collecting TSP~\cite{johnson2000prize}, a problem that is NP-hard, where the profits are given by the dual variables $y$. Therefore, one may ask if (\ref{eq: separation}) can be approximated.

The closest related inapproximability result is that of Feigenbaum, Papadimitriou, and Shenker~\cite{feigenbaum2001sharing} for the \textit{directed} prize-collecting Steiner tree problem. They show, via a gap reduction from SAT, that it is NP-hard to approximate the problem within any constant factor. Intuitively, one would expect their result to carry over to the directed prize-collecting TSP, but it seems hard to modify the gadget in~\cite{feigenbaum2001sharing} for this purpose. The difficulty is that adding the edges necessary to allow tours may enable unanticipated interactions between variables and their negations, which the tree structure avoided. We believe this is an open problem.

The reader may notice that some literature (for example,~\cite{costa2006steiner,chapovska2006variations,nagarajan2008approximation,paul2020budgeted}) mention the result of Feigenbaum et al.~\cite{feigenbaum2001sharing} as an inapproximability result for the \textit{undirected} prize-collecting Steiner tree problem. Although not explicitly mentioned by Feigenbaum et al.~\cite{feigenbaum2001sharing}, their gadget can be easily updated to work on undirected graphs~\cite{paul2020}. The only change needed is to set the profit of clause nodes to $2K$ and the cost of edges between literals and clauses to $K$. The reader may also notice that~\cite{toth2014vehicle,feillet2005traveling} associate constant factor approximation results with the net-worth maximization version of the prize-collecting TSP. However, the references cited therein~\cite{bienstock1993note,goemans1995general} actually correspond to a different version of the prize-collecting TSP, namely the one in which a non-negative penalty is paid for each node excluded from the tour. The simple transformation between the two versions of the problem given in Section 10.2.2 of~\cite{toth2014vehicle}, while valid for optimal solutions, is not approximation preserving. See~\cite{angelelli2014complexity} for some tractable special cases and~\cite{johnson2000prize} for a catalogue of different versions of the prize-collecting Steiner tree problem. 

Since any generalization of the TSP includes the TSP as a special case, an inapproximability result for the directed (or undirected) prize-collecting TSP would also hold for any of the more realistic directed (or undirected) routing problems, outlined in the beginning of this section, that could determine the coefficients $c$. This would already avert the design of an approximate separation oracle\footnote{We say a separation oracle for $\mathcal{P}_v$ is $\alpha$-approximate if, given some $(y, z_v)$, it either certifies $(y, \frac{z_v}{\alpha}) \in \mathcal{P}_v$ or returns a violated constraint. An $\alpha$-approximation algorithm for (\ref{eq: separation}) would yield $\alpha$-approximate separation oracle for $\mathcal{P}_v$. To see this, let $z^*$ be the value obtained by an $\alpha$-approximation algorithm for (\ref{eq: separation}) and let $t^*$ be a trip $t \in T$ achieving $z^*$. If $z^* > z_v$, we know $t^*$ induces a violated constraint. Otherwise, for any $t \in T$ we have $\sum_{r \in t} y_r - c_{tv} \leq \frac{1}{\alpha}\left(\sum_{r \in t^*} y_r - c_{t^*v}\right) = \frac{z^*}{\alpha} \leq \frac{z_v}{\alpha}$ and so $\left(y, \frac{z_v}{\alpha}\right) \in \mathcal{P}_v$.} to obtain a constant factor approximation for the LP relaxation of (\ref{p: rtv}) in the style of~\cite{carr2000randomized,jain2003packing,fleischer2011tight}. Note that this would be in addition to the inherent difficulty with the mixed-sign objective of (\ref{p: drtv}).

In practice, problem (\ref{eq: separation}) appears as the pricing problem within branch-and-price frameworks for vehicle routing, possibly with additional service specific constraints. Pricing problems are typically solved via a variety of exact methods.
See~\cite{toth2014vehicle,feillet2005traveling} for surveys, where they refer it as the \textit{profitable tour problem} and~\cite{riley2019column} for a particularly relevant example. 
See also~\cite{ljubic2006algorithmic, leitner2018dual} for exact branch-and-cut and dual ascent-based branch-and-bound methods for the closely-related net-worth maximizing Steiner tree problem, respectively, and~\cite{ljubic2021solving} for a survey on recent developments.
We emphasize that, for this particular problem, these methods take exponential time in the worst case.

The unfavorable prospect outlined in this section motivates us to assume the candidate trip list $T$ is pre-computed and polynomial-sized whenever we formally consider the LP relaxation of (\ref{p: rtv}). 

\section{Randomized Rounding}
\label{sec: randomized rounding}

We now turn to our randomized rounding algorithm. Consider a feasible instance of the problem and the version of (\ref{p: rtv}) \textit{with} equality constraints. We assume we can solve its LP relaxation (e.g., assuming $|T| = \text{poly}\left(|R|, |V|\right)$ as a formality) to obtain a fractional solution $x$ of value $\text{LP}(\text{\ref{p: rtv}})$. We further assume the cost coefficients $c \geq 0$ are oracle-given. $\alpha$-approximate cost coefficients can be accounted for by Lemma~\ref{lemma: approximate c}. We consider two simple randomized rounding techniques, the second of which overcomes a deal-breaking shortcoming of the first.

In both cases, let $X_{tv} \in \{0, 1\}$ be a random variable indicating the assignment of vehicle $v \in V$ to trip $t \in T$. Let $X$ be the set of $X_{tv}$ random variables for all $(t,v) \in T \times V$. Let $X_v \subseteq X$ be the subset of $X$ involving vehicle $v \in V$ and $X_r \subseteq X$ be the subset of $X$ involving trips $t \in T$ such that $r \in t$.


\subsection{Independent Rounding.} 
\label{sec: independent rounding}

Consider setting each $X_{tv} \in X$ independently to $1$ with probability $0 \leq x_{tv} \leq 1$ and to $0$ otherwise. Clearly $\mathbb{E}\left[c(X)\right] = \sum_{v \in V} \sum_{t \in T} c_{tv} x_{tv} = \text{LP}(\text{\ref{p: rtv}})$. However, the rounded solution may be infeasible. 

One advantage of this technique is that we can easily upper bound the probability of a fixed request $r \in R$ being left unassigned. This is given by
\begin{align*}
    \prod_{(t,v) \in T(r) \times V} (1 - \Pr[X_{tv} = 1])
    &= \prod_{(t,v) \in T(r) \times V} (1 - x_{tv}) \\
    &\leq e^{-\sum_{(t,v) \in T(r) \times V}x_{tv}} \\
    &= e^{-1},
\end{align*}
where the last equality holds by the LP constraints. We will later see how we can handle the event in which a request is over-assigned. The bigger issue is that a vehicle $v \in V$ is over-assigned with non-zero probability (i.e., there are two or more trips $t \in T$ such that $X_{tv} = 1)$. This is problematic because, even if we can upper bound the probability of this event, there is always the possibility that we over-commit our fleet. We could in principle re-sample all random 
variables, potentially many times, until we reach an assignment with no over-commitments. Alternatively, we can bypass this issue altogether through a dependent rounding technique.

\subsection{Dependent Rounding.}
\label{sec: dependent rounding}

Recall each vehicle $v \in V$ satisfies $\sum_{t \in T} x_{tv} = 1$ by the LP constraints. We interpret this as a vehicle-specific probability distribution over the trips. Therefore, independently for each vehicle $v \in V$, assign it to a randomly chosen trip, where the probability of assigning it to trip $t \in T$ is given by $0 \leq x_{tv} \leq 1$. Let $X$ be the set of random variables corresponding to this step and observe that they are no longer independent. 

Note that a request $r \in R$ may appear in multiple trip-vehicle assignments. Allow $r$ to pick one arbitrarily. Here we are crucially using the fact that $T$ is downward closed. Further, by the monotonicity of the cost coefficients $c \geq 0$ with respect to request inclusion, doing this cannot increase the cost of our solution. Define $X'$ analogously to $X$, except this time it corresponds to the final output after the multiplicity correction step. The following result is immediate.

\begin{lemma}
\label{lemma: bound objective}
The expected cost of the final RTV assignment $X'$ is at most $\text{LP}(\text{\ref{p: rtv}})$.
\end{lemma}

Clearly, our procedure ensures each vehicle is assigned to exactly one trip (resolving our previous issue). However, it may still leave some unassigned requests. Nevertheless, we can still bound the probability of this event for any fixed request $r \in R$.

\begin{lemma}
\label{lemma: bound unassigned}
For any request $r \in R$, the probability of it being left unassigned in the final RTV assignment $X'$ is at most $1/e$.
\end{lemma}
\begin{proof}
Since $r$ is assigned in $X$ if and only if it is assigned in $X'$, we can focus on the former. Let $Y_v$ be a binary random variable indicating whether vehicle $v \in V$ is assigned, in $X$, to a trip $t \in T$ such that $r \in t$. Then, $\Pr[Y_v = 1] = \sum_{t \in T(r)} x_{tv}$. 

Note that $r$ is left unassigned in $X$ if and only if $Y_v = 0$ for all $v \in V$. Now, since the rounding is done independently at the vehicle level, the variables $Y_v$ for $v \in V$ are independent, and so $r$ is left unassigned with probability
\begin{align*}
    \prod_{v \in V} (1 - \Pr[Y_v = 1]) 
    &\leq e^{-\sum_{v \in V} \Pr[Y_v = 1]} \\
    &= e^{-\sum_{(t, v) \in T(r) \times V} x_{tv}} 
    = e^{-1},
\end{align*}
where, as before, the last equality holds by the LP constraints.
\end{proof}
\begin{corollary}
\label{cor: cor}
The expected number of requests left unassigned in the final RTV assignment is at most $\frac{|R|}{e}$ (i.e., less than $36.8\%$ of all requests).
\end{corollary}
\begin{proof}
By linearity of expectation.
\end{proof}

Lemma~\ref{lemma: bound objective} and Corollary~\ref{cor: cor}, together with the fact that $\text{LP}(\text{\ref{p: rtv}}) \leq \text{OPT}$, imply Theorem~\ref{thm: main}. 

For any fixed request $r \in R$ we also show that, in the multiplicity correction step, it is unlikely to have too many trips/vehicles to choose from. That is, we bound the probability of $\sum_{(t,v) \in X_r} X_{tv}$ deviating above its unit mean by $\delta \in \mathbb{Z}_{\geq 1}$-many trips. The proof uses a multiplicative Chernoff bound (see for example~\cite{motwani1995randomized}), and it can be found in the Appendix.
\begin{lemma}
\label{lemma: bound multiplicity correction}
For any request $r \in R$, the probability of it being over-assigned by $\delta \in \mathbb{Z}_{\geq 1}$-many trips in the preliminary RTV assignment $X$ is at most $\frac{e^{\delta}}{(1+\delta)^{1 + \delta}}$.
\end{lemma}
We note that we could have similarly used a multiplicative Chernoff bound to produce an alternate proof of Lemma~\ref{lemma: bound unassigned}. Lastly, although we did not need this in our analysis, we note that the sets of random variables $X_v$ for $v \in V$, $X_r$ for $r \in R$, and even $X$ itself, are each negatively associated. We show this in the Appendix for completeness.

We end this section by introducing a class of instances showing that \textit{i}) the integrality gap of the LP relaxation of (\ref{p: rtv}) is at least $2$, and \textit{ii}) the upper bound of $1/e$ on the rejection rate is tight for our algorithm. 

We are given two vehicles of capacity $k > 0$ together with $k+1$ requests. Any trip of size $\leq k$ can be assigned to either vehicle at unit cost (e.g., all requests go from the same origin to the same destination). Figure~\ref{fig: toy example} depicts the shareability graph of our instance when $k=2$. In the $k$th instance, an optimal integer solution has cost $2$. Meanwhile, an optimal fractional solution assigns fractional value $1/\binom{k+1-1}{k-1}$ to each of the $\binom{k+1}{k}$ distinct $k$-passenger trips, each at unit cost. The integrality gap of the $k$th instance is then
\begin{align*}
    \frac{2}{\binom{k+1}{k} \cdot \frac{1+(k-1) \epsilon}{\binom{k+1-1}{k-1}}} 
    = 2 \cdot \frac{k}{k+
    1}.
\end{align*}
Since $\lim_{k \rightarrow \infty} 2 \cdot \frac{k}{k+1} = 2$, the integrality gap of the LP relaxation of (\ref{p: rtv}) is at least $2$. The intuitive reason for this phenomenon is that there is no real distinction between the vehicles. Then, although these instances are trivial in a practical sense, the generality of (\ref{p: rtv}) (which is meant to model distinct vehicles) causes it to miss the obvious structure. 

\begin{figure}[ht]
    \centering
    \includegraphics[width=0.75\columnwidth]{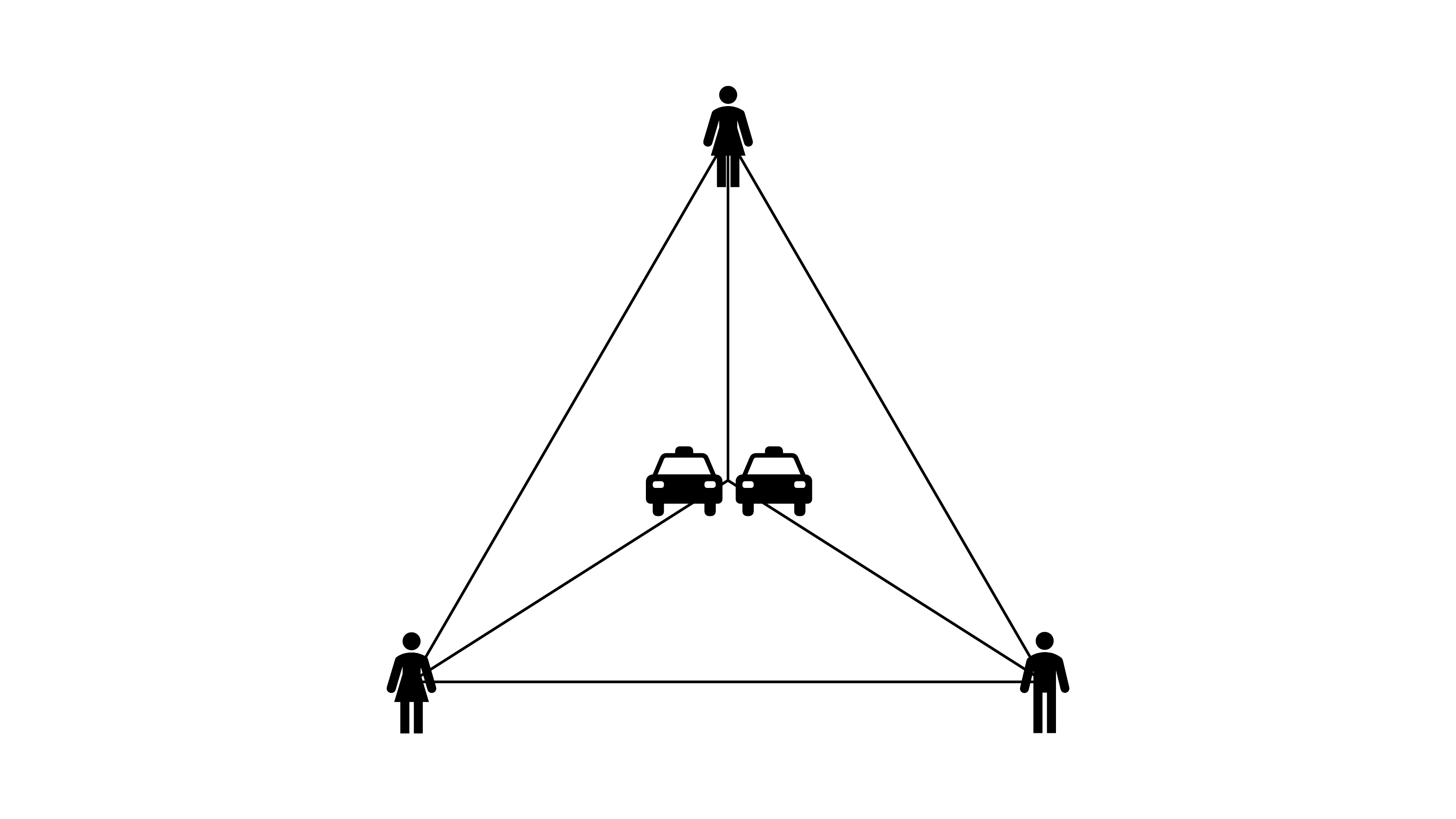}
    \vspace*{-3mm}
    \caption{In this instance, any trip of size 2 or smaller can be served by either vehicle at unit cost. An optimal integer solution has cost $2$, whereas an optimal fractional solution will be half-integral and have cost $3/2$. The integrality gap of this instance is $4/3$.}
    \label{fig: toy example}
\end{figure}

Now we slightly modify our instance class to show the upper bound of $1/e$ on the rejection rate is tight for our algorithm. In the $k$th instance, rather than having $2$ vehicles, we instead have $\binom{k+1}{k} = k+1$ vehicles. Assignment costs do not change, and so the costs of optimal integer or fractional solutions do not change either. However, for fractional solutions, we may reach an (admittedly pathological) optimal solution in which each  vehicle is assigned to exactly one of the $k+1$ distinct $k$-passenger trips with fractional value $1/k$ (and to the empty trip with fractional value $\frac{k-1}{k}$ at zero cost). Then, any fixed request $r \in R$ has uniform fractional support on $k$ distinct vehicles, and so the probability of it being left unassigned by our rounding algorithm is exactly $\left(\frac{k-1}{k}\right)^k$. Since $\lim_{k \rightarrow \infty}\left(\frac{k-1}{k}\right)^k = \frac{1}{e}$, Lemma~\ref{lemma: bound unassigned} is tight in the limit\footnote{We note that if a request $r \in R$ has not necessarily uniform fractional support on $k$ distinct vehicles, we could have used this limit argument together with the inequality of arithmetic and geometric means to obtain yet another proof of Lemma~\ref{lemma: bound unassigned}.}.

\section{Penalty Version and Multiple Rounds}
\label{sec: penalty version and multiple rounds}

The main shortcoming of our algorithm is that it may leave some unassigned requests. Typically, a remedy to this is to repeat the procedure to cover any missed assignments (at the expense of augmenting the expected objective value). However, vehicles cannot be assigned to more than one trip and so we cannot do this in general. Nevertheless, in practice and by design, any unassigned requests are carried over to the next round of batch assignments (e.g., 5 to 30 seconds later). In this section, we study the effects of this design feature on the subsequent assignment of initially unassigned requests.

Pick any request $r \in R$. For now, consider the probability of $r$ being left unassigned after $n$ \textit{feasible} rounds. Let $A_i$ be the event that $r$ is unassigned after $i$ feasible rounds. Then,
\small{
\begin{align*}
    \Pr(A_n) 
    &= \Pr(A_n | A_1)\Pr(A_1) + \Pr(A_n | A_1^c) \Pr(A_1^c) \\
    &= \Pr(A_n | A_1) \Pr(A_1) \leq \frac{1}{e} \Pr(A_n | A_1),
\end{align*}
}
where the second equality holds since $\Pr(A_n | A_1^c) = 0$ and the inequality holds by Lemma~\ref{lemma: bound unassigned}. Likewise, we can write
\small{
\begin{align*}
     \Pr(A_n | A_1) 
     &= \Pr(A_n | A_2, A_1)\Pr(A_2 | A_1) \\
     & \qquad + \Pr(A_n | A_2^c, A_1) \Pr(A_2^c | A_1) \\
     &= \Pr(A_n | A_2, A_1)\Pr(A_2 | A_1) \\
     &\leq \frac{1}{e}\Pr(A_n | A_2, A_1),
\end{align*}
}
where the second equality holds since $\Pr(A_n | A_2^c, A_1) = 0$ and the inequality holds since, given that $r$ is unassigned in the first feasible round, in the second feasible round we attempt to assign it but fail to do so with probability at most $1/e$ by Lemma~\ref{lemma: bound unassigned}. We can extend this argument to $n$ feasible rounds, where $\Pr(A_n | A_n, \cdots, A_1) = 1$, to obtain the following.
\begin{corollary}
\label{cor: high probability}
For any request $r \in R$, the probability of it being left unassigned after $n$ feasible rounds of batch assignments is at most $(1/e)^n$.
\end{corollary}
As in~\cite{alonso2017demand}, we are implicitly accounting for requests that were already picked up by ensuring each term $c_{tv}$ includes the cost of routing passengers inside the vehicle, if any, until the round in which they are dropped off. 

Admittedly, there does not always exist a feasible solution to the RTV assignment problem. For example, this may occur if vehicle supply is low relative to travel demand. To formally handle this, one can introduce a dummy vehicle $v_r$ for each request $r \in R$. This vehicle can only be assigned to trip $\{r\} \in T$, in which case we incur a large cost $\kappa_r \geq 0$ representing the penalty for ignoring $r$, or the empty trip $\emptyset \in T$ at zero cost. This is in fact the strategy followed in practice.

By doing this, we actually ensure we meet the feasibility assumption required by Theorem~\ref{thm: main}. We moreover preserve the monotonicity assumption required on the cost coefficients. Therefore, we can again use our randomized algorithm, although our performance guarantee is now with respect to the modified objective function which now includes the penalty terms. Still, a request $r \in R$ is unassigned (in terms of ILP feasibility) with probability at most $1/e$. In this case, we can formally cover $r$ by assigning it to its dummy vehicle $v_r$ and paying the corresponding penalty $\kappa_r$. This produces a feasible ILP solution yielding the following corollary.
\begin{corollary}
Suppose $T$ is polynomial-sized. If trip-vehicle costs are oracle-given and monotonic increasing w.r.t. request inclusion, and there is a penalty $\kappa_r \geq 0$ for ignoring request $r \in R$, there is a polynomial time randomized algorithm
yielding a solution of expected cost at most $\text{OPT} + \frac{1}{e} \sum_{r \in R} \kappa_r$, where $\text{OPT}$ is the optimal feasible sum of trip-vehicle assignment costs and penalties. If trip-vehicle assignment costs are $\alpha$-approximate, we pay an additional factor of $\alpha$ in the first term of the expected cost.
\end{corollary}
In practice, if there is enough vehicle supply to clear demand on every round  (intuitively, the average number of vehicle seats available is greater than or equal to the request arrival rate times the average time spent in the system~\cite{little1961proof}), one can prioritize unassigned requests by augmenting their penalty terms in subsequent rounds of batch assignments\footnote{If there is not enough vehicle supply, unassigned request queue up and quality of service is unmet, leading to reneging customers.}.

We now interpret Corollary~\ref{cor: high probability} within the context of the penalty version of the problem.
\begin{remark}
Corollary~\ref{cor: high probability} is applicable to the penalty version of the problem in the following sense: For any request $r \in R$, with high probability over a number of rounds, the algorithm is not forced to ignore $r$, but rather assigns it to a vehicle or purposely ignores it, in accordance with the sequence of LP solutions.
\end{remark}
That is to say, if a request is ignored by our algorithm after a number of rounds, it is increasingly likely to be so because of the LP solutions and not because of a rounding error.

Lastly, we note than when the penalty terms eclipse the assignment costs (e.g., routing costs), the penalty version of the problem is essentially equivalent to a maximum coverage version of the problem, where the objective is to serve as many requests as possible given the available fleet. This maximum coverage problem is still subject to the quality of service constraints that are implicit in the RTV graph, and therefore can be posed as an instance of the cardinality version of the maximum coverage problem with group budget constraints, for which a $1/e$-approximation algorithm is known~\cite{chekuri2004maximum}, assuming $T$ is given explicitly. Our randomized algorithm matches this guarantee in the sense that the expected fraction of requests that could have been covered but were not is at most $1/e$. If assignment costs are non-zero, our method provides the additional benefit that the quality of the assignment (e.g., in terms of the routing costs of assigned requests) is also optimized as part of the LP.

\section{Computational Experiments}
\label{sec: computational experiments}

We now share our computational experience. To evaluate our algorithm, we first run a number of simulations of a day in the operations of a high-capacity ridesharing system. We use publicly available NYC Taxi and Limousine Commission (TLC) data~\cite{taxidata} and an implementation of the Alonso-Mora et al. framework~\cite{alonso2017demand}. In each simulation, we solve the penalty version of the distance-minimizing RTV assignment problem using a commercial ILP solver. Next, we gather all instances of the RTV assignment problem solved throughout our simulations and solve them once again using our techniques. Note that we are \textit{not} comparing simulation paths to evaluate system-level performance. Rather, we use simulation paths to produce sets of test instances. This way we can have a side-by-side comparison of our algorithm against the use of a commercial ILP solver (the same we use to solve our LP relaxations). Each simulation produces $1,440$ instances. 

We run a total of 10 simulations, one for each combination\footnote{For vehicle capacities $k \geq 4$, our simulations employ time-outs in the trip generation step.} of $500$ or $1000$ vehicles and vehicles of capacity $k = 2, 3, 4, 5, 6$. For simulations with $500$ vehicles, the mean number of requests per instance is around $370$ requests\footnote{Even with a fixed number of vehicles, different capacities $k$ yield different simulation paths (e.g., with different service/reneging rates), and so the instances are not identical.}. Likewise, for simulations with $1000$ vehicles, the mean number of requests per instance is around $560$ requests.

In addition, we evaluate a deterministic rounding heuristic inspired by our randomized rounding algorithm. The only difference between the two is that, in our heuristic, each vehicle $v \in V$ is deterministically assigned to the trip $t^* \in T$ with largest fractional value (i.e., vehicle $v \in V$ is assigned to trip $t^* = \argmax_{t \in T}\{x_{tv}\}$). Note that we may still need to execute the multiplicity correction step.

Table~\ref{tab: rejections} presents statistics on the percentage of rejected requests for each method. As expected, solving the ILP leads to the lowest rate of request rejection. Moreover, the rejection rate is monotonic decreasing in the number of vehicles and in the vehicle capacity $k$. Surprisingly, our rounding algorithms are worse than the ILP on this metric by no more than one percentage point, and often much less. For our randomized rounding algorithm, this is far better than the theoretical guarantee provided in Theorem~\ref{thm: main}. We explore this further toward the end of this section. We also note that our deterministic heuristic slightly outperforms our randomized algorithm on this metric. Since the rejection rates are so close, the assignment costs (in terms of vehicle kilometers traveled) is comparable for all three methods. Solving the ILP leads to a slightly higher number of vehicle kilometers traveled since slightly less requests are rejected. We summarize the corresponding statistics in Table~\ref{tab: assignment costs}, which can be found in the Appendix.

\begin{table}[ht]
\centering{
\caption{Request rejection.}
\vspace{0.2cm}
\label{tab: rejections}
\resizebox{0.85\columnwidth}{!}{%
\begin{tabular}{clrrr}
500 Veh.             & \multicolumn{1}{c}{} & \multicolumn{3}{c}{Requests Rejected {[}\%{]}}                                           \\ \cline{3-5} 
$k$                  & \multicolumn{1}{c}{} & \multicolumn{1}{c}{ILP} & \multicolumn{1}{c}{Rand. Rnd.} & \multicolumn{1}{c}{Det. Rnd.} \\ \hline
\multirow{2}{*}{2}   & Mean                 & 27.09                   & 27.37                          & 27.32                         \\
                     & Med                  & 31.50                   & 31.55                          & 31.56                         \\
\multirow{2}{*}{3}   & Mean                 & 20.31                   & 20.73                          & 20.65                         \\
                     & Med                  & 23.83                   & 23.98                          & 24.00                         \\
\multirow{2}{*}{4}   & Mean                 & 15.41                   & 15.84                          & 15.77                         \\
                     & Med                  & 17.82                   & 18.12                          & 18.12                         \\
\multirow{2}{*}{5}   & Mean                 & 12.69                   & 13.14                          & 13.07                         \\
                     & Med                  & 14.50                   & 14.88                          & 14.80                         \\
\multirow{2}{*}{6}   & Mean                 & 11.51                   & 11.93                          & 11.85                         \\
                     & Med                  & 12.91                   & 13.29                          & 13.20                         \\ \hline
\multicolumn{1}{l}{} &                      & \multicolumn{1}{l}{}    & \multicolumn{1}{l}{}           & \multicolumn{1}{l}{}          \\
1000 Veh.            & \multicolumn{1}{c}{} & \multicolumn{3}{c}{Requests Rejected {[}\%{]}}                                           \\ \cline{3-5} 
$k$                  & \multicolumn{1}{c}{} & \multicolumn{1}{c}{ILP} & \multicolumn{1}{c}{Rand. Rnd.} & \multicolumn{1}{c}{Det. Rnd.} \\ \hline
\multirow{2}{*}{2}   & Mean                 & 13.06                   & 13.58                          & 13.53                         \\
                     & Med                  & 13.33                   & 13.76                          & 13.69                         \\
\multirow{2}{*}{3}   & Mean                 & 8.84                    & 9.74                           & 9.57                          \\
                     & Med                  & 8.36                    & 9.24                           & 9.09                          \\
\multirow{2}{*}{4}   & Mean                 & 6.12                    & 7.11                           & 6.93                          \\
                     & Med                  & 5.39                    & 6.43                           & 6.24                          \\
\multirow{2}{*}{5}   & Mean                 & 4.69                    & 5.61                           & 5.46                          \\
                     & Med                  & 3.95                    & 4.88                           & 4.77                          \\
\multirow{2}{*}{6}   & Mean                 & 4.15                    & 5.02                           & 4.90                          \\
                     & Med                  & 3.44                    & 4.44                           & 4.31                          \\ \hline
\end{tabular}%
}}
\end{table}

Table~\ref{tab: computation time} summarizes the performance of our implementation of all three methods in terms of computation time. We observe that both rounding methods are faster than solving the ILP to optimality, with the percent improvement for mean and median values ranging between $2-7\%$. Importantly, we observe that the percent improvement for mean values is consistently higher than the percent improvement for median values. This shows the computation time distributions are skewed to the right, and so the percent improvement is higher for the worst case (i.e., slowest) instances. Arguably, these are the critical instances in real-time decision making, and so LP rounding methods provide a way to improve on the worst case computation time performance.

\begin{table}[ht]
\centering
\caption{Computation time.}
\vspace{0.2cm}
\label{tab: computation time}
\resizebox{0.9\columnwidth}{!}{%
\begin{tabular}{clrrrrr}
500 Veh. &
  \multicolumn{1}{c}{} &
  \multicolumn{1}{c}{ILP} &
  \multicolumn{2}{c}{Rand. Rnd.} &
  \multicolumn{2}{c}{Det. Rnd.} \\ \cline{3-7} 
$k$ &
  \multicolumn{1}{c}{} &
  \multicolumn{1}{c}{\begin{tabular}[c]{@{}c@{}}Time\\ {[}s{]}\end{tabular}} &
  \multicolumn{1}{c}{\begin{tabular}[c]{@{}c@{}}Time\\ {[}s{]}\end{tabular}} &
  \multicolumn{1}{c}{\% Diff} &
  \multicolumn{1}{c}{\begin{tabular}[c]{@{}c@{}}Time\\ {[}s{]}\end{tabular}} &
  \multicolumn{1}{c}{\% Diff} \\ \hline
\multirow{2}{*}{2} &
  Mean &
  0.152 &
  0.154 &
  0.97 &
  0.144 &
  -5.01 \\
 &
  Med &
  0.154 &
  0.158 &
  2.59 &
  0.150 &
  -2.91 \\
\multirow{2}{*}{3} &
  Mean &
  0.229 &
  0.222 &
  -2.99 &
  0.215 &
  -6.26 \\
 &
  Med &
  0.172 &
  0.175 &
  2.12 &
  0.168 &
  -2.03 \\
\multirow{2}{*}{4} &
  Mean &
  0.319 &
  0.303 &
  -4.77 &
  0.297 &
  -6.91 \\
 &
  Med &
  0.199 &
  0.197 &
  -1.19 &
  0.189 &
  -5.42 \\
\multirow{2}{*}{5} &
  Mean &
  0.324 &
  0.309 &
  -4.66 &
  0.302 &
  -6.86 \\
 &
  Med &
  0.215 &
  0.212 &
  -1.80 &
  0.203 &
  -5.69 \\
\multirow{2}{*}{6} &
  Mean &
  0.310 &
  0.297 &
  -4.26 &
  0.289 &
  -6.85 \\
 &
  Med &
  0.217 &
  0.211 &
  -2.90 &
  0.203 &
  -6.64 \\ \hline
\multicolumn{1}{l}{} &
   &
  \multicolumn{1}{l}{} &
  \multicolumn{1}{l}{} &
  \multicolumn{1}{l}{} &
  \multicolumn{1}{l}{} &
  \multicolumn{1}{l}{} \\
1000 Veh. &
  \multicolumn{1}{c}{} &
  \multicolumn{1}{c}{ILP} &
  \multicolumn{2}{c}{Rand. Rnd.} &
  \multicolumn{2}{c}{Det. Rnd.} \\ \cline{3-7} 
$k$ &
  \multicolumn{1}{c}{} &
  \multicolumn{1}{c}{\begin{tabular}[c]{@{}c@{}}Time\\ {[}s{]}\end{tabular}} &
  \multicolumn{1}{c}{\begin{tabular}[c]{@{}c@{}}Time\\ {[}s{]}\end{tabular}} &
  \multicolumn{1}{c}{\% Diff} &
  \multicolumn{1}{c}{\begin{tabular}[c]{@{}c@{}}Time\\ {[}s{]}\end{tabular}} &
  \multicolumn{1}{c}{\% Diff} \\ \hline
\multirow{2}{*}{2} &
  Mean &
  0.731 &
  0.710 &
  -2.82 &
  0.697 &
  -4.64 \\
 &
  Med &
  0.761 &
  0.756 &
  -0.63 &
  0.739 &
  -2.90 \\
\multirow{2}{*}{3} &
  Mean &
  1.259 &
  1.188 &
  -5.65 &
  1.178 &
  -6.47 \\
 &
  Med &
  0.996 &
  0.959 &
  -3.76 &
  0.951 &
  -4.50 \\
\multirow{2}{*}{4} &
  Mean &
  2.075 &
  1.942 &
  -6.39 &
  1.935 &
  -6.74 \\
 &
  Med &
  1.472 &
  1.411 &
  -4.14 &
  1.406 &
  -4.47 \\
\multirow{2}{*}{5} &
  Mean &
  2.212 &
  2.073 &
  -6.26 &
  2.067 &
  -6.54 \\
 &
  Med &
  1.708 &
  1.640 &
  -3.96 &
  1.625 &
  -4.88 \\
\multirow{2}{*}{6} &
  Mean &
  2.124 &
  1.994 &
  -6.11 &
  1.984 &
  -6.60 \\
 &
  Med &
  1.710 &
  1.639 &
  -4.15 &
  1.631 &
  -4.61 \\ \hline
\end{tabular}%
}
\end{table}




Lastly, we address the question of why the rejection rate of our rounding algorithms is so close to that achieved by the ILP. Figure~\ref{fig: full lp support} shows a histogram of the fractional values supported on the LP solutions across all $1,440$ instances of a simulation with $1000$ vehicles of capacity $k=4$. We see that the overwhelming majority ($\sim 96\%$) of non-zero fractional values are in fact integral. Of course, the rounding algorithms make no mistakes on these assignments. We extend this observation in Figure~\ref{fig: non integral lp support}, which shows a histogram of the fractional values strictly between $0$ and $1$. We observe that whenever an assignment is not integral, it is likely to be half-integral ($\sim 62\%$ of non-integral variables are half-integral). In half-integral cases, we can a posteriori tighten the probability that our randomized rounding algorithms ignores a fixed request $r \in R$ to be $\leq 1/4$.

\begin{figure}[!ht]
    \centering
    \includegraphics[width=0.75\columnwidth]{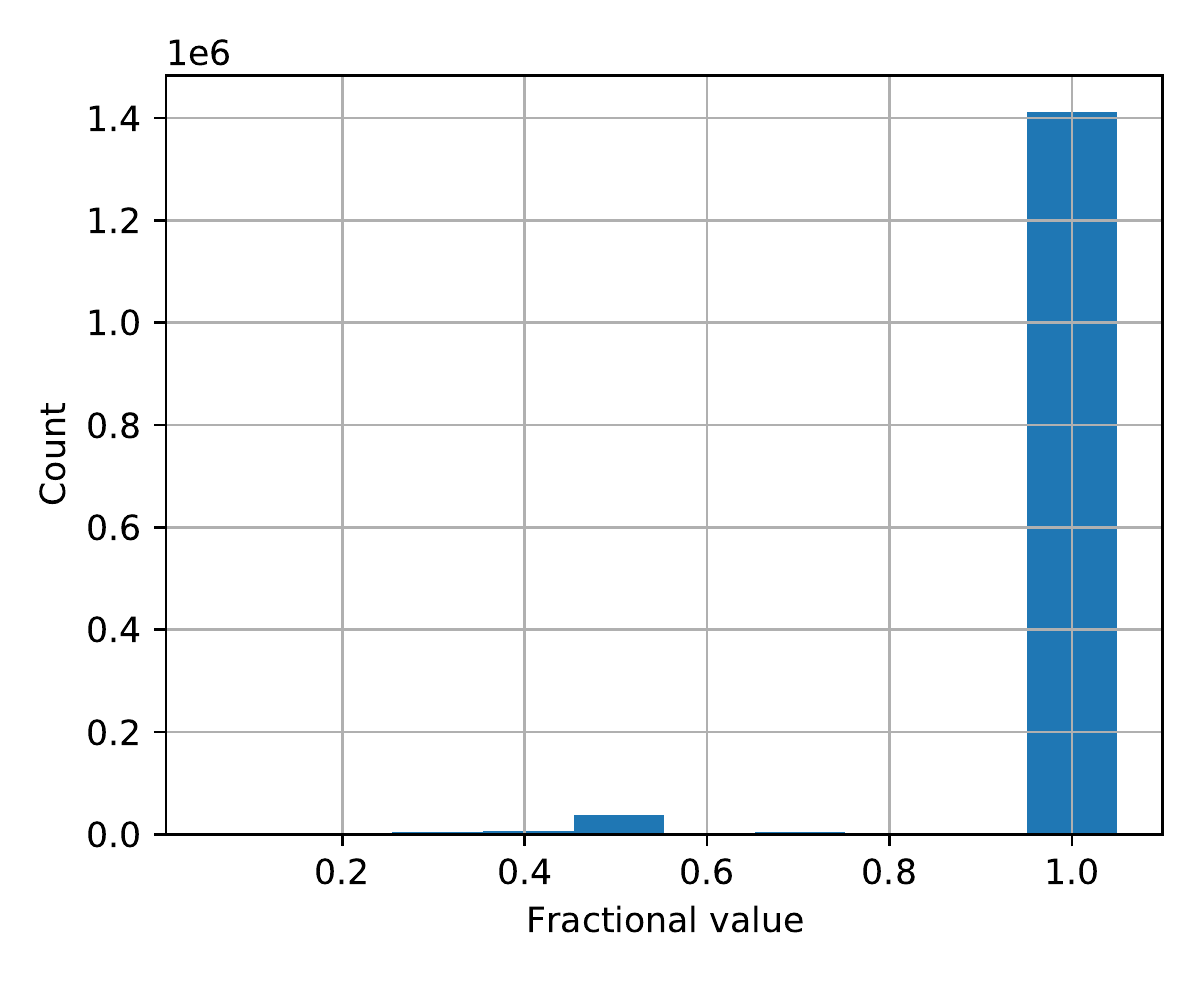}
    \vspace*{-0.7cm}
    \caption{Histogram of LP support.}
    \label{fig: full lp support}
\end{figure}
\begin{figure}[ht]
    \centering
    \includegraphics[width=0.75\columnwidth]{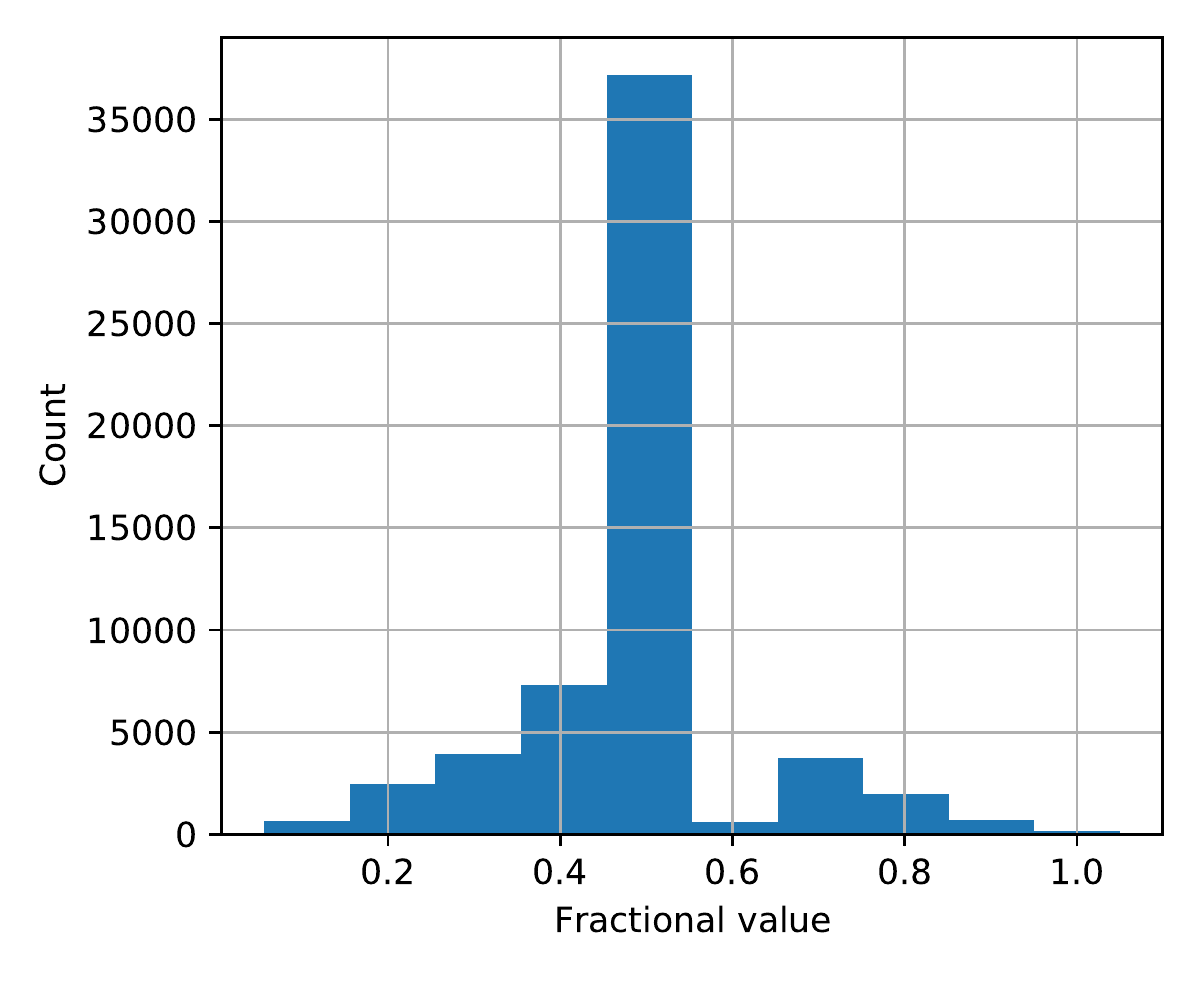}
    \vspace*{-0.7cm}
    \caption{Histogram of non-integral LP support.}
    \label{fig: non integral lp support}
\end{figure}

We have been unable to concretely characterize reasons why the LP relaxation of (\ref{p: rtv}) is so tight. We can produce toy examples as the one in Figure~\ref{fig: toy example} (and more generally the class of instances described in Section~\ref{sec: dependent rounding}). However, these examples are neither realistic nor exhaustive, and so we believe tackling this problem (e.g., in a data-driven way, with data sets besides the NYC TLC data~\cite{taxidata}) is an interesting research direction. In particular, we do not know whether the LP is \textit{always} tight in practice. If there are practical instances in which the LP is not as tight, they will likely be harder for an ILP solver, and so they might more strongly showcase the advantage of LP rounding.

\section*{Acknowledgments}
We thank Matthew Zalesak for sharing his implementation of the Alonso-Mora et al.~\cite{alonso2017demand} framework with us, and the anonymous reviewers for helpful feedback.

\bibliographystyle{siam}
\bibliography{bib}

\section*{Appendix}

\subsection*{Relation to Exact Cover and Set-Partitioning.}
\label{sec: relation} 
Given a pair $(X, \mathcal{S)}$ where $X$ is a ground set and $\mathcal{S}$ is a collection of subsets of the ground set, an \textit{exact cover} is a sub-collection of $\mathcal{S}$ that partitions $X$. The \textit{exact cover problem} asks whether $(X, \mathcal{S)}$ has an exact cover, and it is well-known to be NP-complete~\cite{karp1972reducibility}. Assuming the trip set $T$ is given explicitly, we can pose the feasibility version of the RTV assignment problem as an instance of the exact cover problem. To see this, let $X = R \cup V$ be the ground set and $\mathcal{S} = \{t \cup \{v\} : t \in T, v \in V(t)\}$ be the collection of subsets of the ground set. 

The weighted \textit{set-partitioning problem} (i.e., set cover with equality constraints)~\cite{garfinkel1969set,balas1976set} is an optimization problem closely related to the exact cover problem. Given non-negative cost $c_S \geq 0$ for each $S \in \mathcal{S}$, the problem asks for a minimum cost sub-collection of $\mathcal{S}$ that partitions $X$. An instance of the RTV assignment problem can be seen as a special instance of the weighted set partitioning problem. In abstract terms:
\begin{enumerate}
    \item There exists a subset $Y \subseteq X$ of the ground set such that the sub-collections $\mathcal{S}(y) = \{S \in \mathcal{S} : y \in S\}$ for $y \in Y$ partition $\mathcal{S}$. 
    \item For each $y \in Y$, the contraction $\mathcal{S}(y) / y = \{S \subseteq X - y: S + y \in \mathcal{S}(y)\}$ is downward closed (i.e., if $A \in \mathcal{S}(y) / y$ and $B \subseteq A$, then $B \in \mathcal{S}(y) / y$ ).
\end{enumerate}
Indeed, we satisfy these properties by again letting $X = R \cup V$ be the ground set, $\mathcal{S} = \{t \cup \{v\} : t \in T, v \in V(t)\}$ be the collection of subsets of the ground set, and $Y = V$. In our work we require trip-vehicle assignment costs to be monotonic increasing with respect to request inclusion. This is to say that the cost function $c: \mathcal{S} \rightarrow \mathbb{R}_{\geq 0}$ is monotone increasing with respect to inclusion in $X \setminus Y$. 

A lighthearted example in this class of problems is a wedding table planning problem where guests may only share a table if they mutually know each other. Then, $R$ is the set of guests and $V$ is the set of tables. While we state our results in the language of the RTV assignment problem, they more generally apply to this class of instances of set-partitioning. 

\subsection*{Proof of Lemma~\ref{lemma: bound multiplicity correction}}
\begin{proof}
Again, let $Y_v$ be a binary random variable indicating whether vehicle $v \in V$ is assigned, in $X$, to a trip $t \in T$ such that $r \in t$. Let $Y = \sum_{v \in V} Y_v$ (i.e., the number of different vehicles $r$ is assigned to in $X$). Since the rounding is done independently at the vehicle level, the variables $Y_v$ for $v \in V$ are independent, and so we can apply a multiplicative Chernoff bound. Namely, for any $\delta > 0$ we have
\begin{align*}
    \Pr[Y > (1 + \delta)E[Y]] \leq \left(\frac{e^{\delta}}{(1+\delta)^{1 + \delta}} \right)^{E[Y]}.
\end{align*}
We now show that $E[Y] = 1$. Since for each $v \in V$ we have $\sum_{X_{tv} \in X_v} X_{tv} \leq 1$ always, $Y$ in fact corresponds to the number of different trips $r$ is assigned to in $X$. That is, $Y = \sum_{X_{tv} \in X_r} X_{tv}$ and so 
\begin{align*}
    E\left[Y\right] 
    &= E\left[\sum_{X_{tv} \in X_r} X_{tv}\right] = \sum_{(t,v) \in T(r) \times V} E[X_{tv}] \\
    &= \sum_{(t,v) \in T(r) \times V} x_{tv} = 1,
\end{align*}
where, as before, the last equality holds by the LP constraints. Given this, we can interpret $\delta \in \mathbb{Z}_{\geq 1}$ as the number of trips $r$ is over-assigned by in $X$.
\end{proof}

\subsection*{Negative Association.}

A set of random variables $X_1, \cdots, X_n$ is said to be negatively associated if, for every two disjoint index sets $I, J \subseteq [n]$ and every two functions $f: \mathbb{R}^{|I|} \rightarrow \mathbb{R}$ and $g: \mathbb{R}^{|J|} \rightarrow \mathbb{R}$, both non-decreasing or both non-increasing, we have $\mathbb{E}[f(X_i, i \in I)g(X_j, j \in J)]  \leq \mathbb{E}[f(X_i, i \in I)] \mathbb{E}[g(X_j, j \in J)]$.

The significance of negative association is that, while negatively associated random variables may be dependent, they are so in a way that exhibits concentration of measure Although we did not need negative association to analyze our dependent randomized rounding algorithm, we show it for completeness in some relevant sets of random variables.




Note that for any vehicle $v \in V$, the random variables within $X_{v}$ are \textit{not} independent. However, the following lemma implies they are negatively associated.

\begin{lemma}[~\cite{dubhashi1996balls}]
Let $X_1, \cdots, X_n$ be zero-one random variables with $\sum_{i=1}^{n} X_i = 1$ always. Then, $X_1, \cdots, X_n$ are negatively associated.
\end{lemma}

Moreover, note that since each vehicle is treated independently, the sets $X_{v}$ for all $v \in V$ are independent from one another. Couple that with the following.
\begin{lemma}[Closure Properties~\cite{joag1983negative}]
\
\begin{enumerate}
    \item The union of independent sets of negatively associated random variables is negatively associated. 
    \item A subset of two or more negatively associated random variables is negatively associated.
\end{enumerate}
\end{lemma}
We immediately see that, in our algorithm, $X$ is itself negatively associated. Moreover, for each $r \in R$, we see that $X_r \subseteq X$ is negatively associated.



\subsection*{Assignment Costs.} See Table~\ref{tab: assignment costs}.

\begin{table}[ht]
\centering
\caption{Assignment costs.}
\vspace{0.2cm}
\label{tab: assignment costs}
\resizebox{\columnwidth}{!}{%
\begin{tabular}{clrrr}
500 Veh.             & \multicolumn{1}{c}{} & \multicolumn{3}{c}{Distance Traveled {[}km{]}}                                           \\ \cline{3-5} 
$k$                  & \multicolumn{1}{c}{} & \multicolumn{1}{c}{ILP} & \multicolumn{1}{c}{Rand. Rnd.} & \multicolumn{1}{c}{Det. Rnd.} \\ \hline
\multirow{2}{*}{2}   & Mean                 & 395.61                  & 395.30                         & 395.33                        \\
                     & Med                  & 429.74                  & 429.57                         & 429.58                        \\
\multirow{2}{*}{3}   & Mean                 & 438.05                  & 437.61                         & 437.65                        \\
                     & Med                  & 488.23                  & 488.11                         & 488.05                        \\
\multirow{2}{*}{4}   & Mean                 & 451.09                  & 450.61                         & 450.64                        \\
                     & Med                  & 506.58                  & 506.03                         & 506.15                        \\
\multirow{2}{*}{5}   & Mean                 & 458.21                  & 457.71                         & 457.75                        \\
                     & Med                  & 514.84                  & 514.29                         & 514.46                        \\
\multirow{2}{*}{6}   & Mean                 & 459.75                  & 459.26                         & 459.32                        \\
                     & Med                  & 518.64                  & 517.75                         & 517.86                        \\ \hline
\multicolumn{1}{l}{} &                      & \multicolumn{1}{l}{}    & \multicolumn{1}{l}{}           & \multicolumn{1}{l}{}          \\
1000 Veh.            & \multicolumn{1}{c}{} & \multicolumn{3}{c}{Distance Traveled {[}km{]}}                                           \\ \cline{3-5} 
$k$                  & \multicolumn{1}{c}{} & \multicolumn{1}{c}{ILP} & \multicolumn{1}{c}{Rand. Rnd.} & \multicolumn{1}{c}{Det. Rnd.} \\ \hline
\multirow{2}{*}{2}   & Mean                 & 712.89                  & 711.84                         & 711.94                        \\
                     & Med                  & 843.35                  & 842.72                         & 842.62                        \\
\multirow{2}{*}{3}   & Mean                 & 754.78                  & 752.95                         & 753.17                        \\
                     & Med                  & 900.68                  & 898.66                         & 899.36                        \\
\multirow{2}{*}{4}   & Mean                 & 755.05                  & 753.04                         & 753.22                        \\
                     & Med                  & 901.22                  & 899.39                         & 899.78                        \\
\multirow{2}{*}{5}   & Mean                 & 750.11                  & 748.11                         & 748.29                        \\
                     & Med                  & 897.48                  & 895.57                         & 894.98                        \\
\multirow{2}{*}{6}   & Mean                 & 746.58                  & 744.65                         & 744.76                        \\
                     & Med                  & 894.76                  & 892.65                         & 893.01                        \\ \hline
\end{tabular}%
}
\end{table}

\end{document}